\documentclass[11pt]{article}

\usepackage[numbers]{natbib}
\usepackage{amsmath}
\usepackage{amssymb}
\usepackage{amsthm}
\usepackage{xcolor}
\usepackage{hyperref}
\usepackage{enumitem}
\usepackage[margin=1in]{geometry}

\hypersetup{colorlinks=true,citecolor=blue,linkcolor=blue}
\allowdisplaybreaks

\newcommand{\paperTitle}{Improved RIP Bounds for Gaussian Partial Circulant Matrices}
\newcommand{\paperAuthor}{Zhao Song\thanks{\texttt{magic.linuxkde@gmail.com}.}}

\DeclareMathOperator{\E}{\mathbb{E}}
\renewcommand{\d}{\mathrm{d}}

\theoremstyle{plain}
\newtheorem{theorem}{Theorem}[section]

\newtheorem{definition}[theorem]{Definition}

\theoremstyle{definition}

\begin{document}

\date{}
\title{\paperTitle}
\author{\paperAuthor}
\maketitle

\begin{abstract}
We prove an improved restricted isometry bound for Gaussian partial circulant
matrices with arbitrary prescribed sampling sets.  There is a universal constant
$C>0$ such that the following holds.  Let $1\leq K\leq m\leq N$ be positive
integers, let $\Omega\subset\mathbb Z_N$ be any fixed set with $|\Omega|=m$,
and let $g\sim\mathcal N(0,I_N)$.  For every $\delta,\eta\in(0,1)$, the
normalized partial circulant matrix generated by $g$ has the RIP of order $K$
with constant at most $\delta$, with probability at least $1-\eta$ over the
draw of $g$, provided
\[
    m\geq C\delta^{-2}K
    \max\{\log^2(eK)\log(2N)\log(em),\log(2/\eta)\}.
\]
The proof refines the Maurey entropy step in the chaos-process argument by
combining a noncommutative Khintchine inequality with a Schatten moment
estimate controlled by $m$, replacing one factor $\log(2N)$ in the
Krahmer--Mendelson--Rauhut bound by $\log(em)$.

\end{abstract}


\section{Introduction}

Cand{\`e}s and Tao~\cite[Definition~1.1]{ct05} introduced the restricted
isometry constants and the sparse near-isometry condition now known as the
restricted isometry property (RIP), which requires a linear map to
approximately preserve the Euclidean norm of every sparse vector.  RIP is a
canonical uniform embedding property.  For matrices with independent
subgaussian entries, the benchmark row complexity is
$m\asymp K\log(eN/K)$~\cite{mpt08,bdd08}, up to the dependence on the
distortion and failure probability.  A central question is the extent to which
this scaling persists under strong algebraic dependence.

Gaussian partial random circulant matrices form a particularly rigid test
case.  A single vector $g\sim\mathcal N(0,I_N)$ generates the convolution
operator $C_g$, which is then restricted to a prescribed set
$\Omega\subset\mathbb Z_N$ of $m$ coordinates.  Consequently, the sampled
rows are dependent shifts of a single Gaussian vector.  Earlier uniform RIP
estimates for arbitrary fixed row sets had a worse dependence on the sparsity
level.  Rauhut~\cite[Corollary~III.3]{rau09} obtained a quadratic-in-$K$
bound, and Rauhut--Romberg--Tropp~\cite[Theorems~1.1 and~1.2]{rrt12} reduced
the leading dependence to $K^{3/2}$; they also observed that their expectation
estimate applies to Gaussian generators~\cite[Section~1.3]{rrt12}.  Using
generic chaining, Krahmer--Mendelson--Rauhut~\cite[Theorem~4.1]{kmr14}
reduced the polynomial dependence to linear in $K$, up to logarithmic factors.
For arbitrary deterministic sampling sets, and with the failure-probability
term suppressed, their sample condition is
$m\gtrsim\delta^{-2}K\log^2(eK)\log^2(2N)$.

We show that one ambient-dimension logarithm can be replaced by the
measurement-dependent factor $\log(em)$.  Specifically, we establish the
sufficient condition
\[
    m\gtrsim\delta^{-2}K\log^2(eK)\log(2N)\log(em).
\]
This yields a strict
asymptotic improvement when $\log m=o(\log N)$; when $m$ is polynomial in
$N$, we do not claim an improvement in logarithmic order.

The key ingredient is a refined coarse-scale entropy estimate.  Rather than
first dominating the operator metric by a Fourier $\ell_\infty$ norm and
incurring the cost of a maximum over $N$ frequencies, we work directly with
the $m\times N$ partial-shift matrices and apply a noncommutative Khintchine
inequality in Schatten classes.  The resulting Schatten moment bound depends
on the number of rows $m$ rather than the ambient dimension, so the Rademacher
estimate contributes
$\sqrt{\log(em)}$ rather than $\sqrt{\log N}$.  The subsequent Maurey argument
squares this estimate and produces the desired $\log(em)$ factor; the remaining
chaining argument follows the chaos-process framework of~\cite{kmr14}.

\subsection{Our results}

\begin{definition}[Restricted isometry property]
\label{def:rip}
\normalfont
For $A\in\mathbb R^{m\times N}$, define its restricted isometry constant of
order $K$ by
\[
    \delta_K(A):=\sup_{\substack{x\in\mathbb R^N,\ \|x\|_2=1\\ \|x\|_0\leq K}} |\|Ax\|_2^2-1|.
\]
We say that $A$ satisfies the restricted isometry property of order $K$ with
constant $\delta$ if $\delta_K(A)\leq\delta$.
\end{definition}

\begin{definition}[Partial random circulant matrix]
\label{def:partial_random_circulant}
\normalfont
Let $g\sim\mathcal N(0,I_N)$, and let $C_g\in\mathbb R^{N\times N}$ be the
circulant matrix defined by $C_gx=g*x$.  Here $*$ denotes the unnormalized
circular convolution $(g*x)_r:=\sum_{j\in\mathbb Z_N}g_{r-j}x_j$ on
$\mathbb Z_N$, with indices taken modulo $N$.  For a fixed set
$\Omega\subset\mathbb Z_N$ with $|\Omega|=m$, let
$R_\Omega:\mathbb R^N\to\mathbb R^m$ denote the coordinate-restriction
operator, with the coordinates indexed by $\Omega$ listed in any fixed order,
and define
\[
    A_\Omega(g):=\frac{1}{\sqrt m}R_\Omega C_g\in\mathbb R^{m\times N}.
\]
Following Krahmer--Mendelson--Rauhut~\cite[Sections~1.2 and~4]{kmr14}, we
call $A_\Omega(g)$ a partial random circulant matrix generated by $g$.
Throughout, $\Omega$ is fixed, and the only randomness is in $g$.
\end{definition}

\begin{theorem}[Main result, Gaussian partial circulant RIP]
\label{thm:gaussian_upper}
There is a universal constant $C>0$ such that the following holds.  Let
$1\leq K\leq m\leq N$ be positive integers, let
$\Omega\subset\mathbb Z_N$ be any fixed set of $m$ distinct coordinates, and
let $g\sim\mathcal N(0,I_N)$.  Then, for every $\delta,\eta\in(0,1)$,
\[
    \Pr_g[\delta_K(A_\Omega(g))\leq\delta]
    \geq 1-\eta
\]
whenever
\begin{equation}
\label{eq:gaussian_sample_implicit}
    m
    \geq
    C\delta^{-2}K
    \max\{
        \log^2(eK)\log(2N)\log(em),
        \log(2/\eta)
    \}.
\end{equation}
\end{theorem}

\paragraph{Improvement over KMR.}
For arbitrary fixed row sets, Krahmer--Mendelson--Rauhut~\cite[Theorem~4.1]{kmr14}
prove RIP under the condition
\[
    m
    \gtrsim
    \delta^{-2}K
    \max\{
        \log^2(eK)\log^2(2N),
        \log(2/\eta)
    \}.
\]
In the first term of this condition, Theorem~\ref{thm:gaussian_upper}
replaces one factor $\log(2N)$ by $\log(em)$.  For fixed $\delta$ and $\eta$,
taking $K=O(1)$ improves the sufficient measurement bound from
$m=O(\log^2N)$ to $m=O(\log N\log\log N)$.  More generally, if $\delta$ and
$\eta$ are fixed and $K=(\log N)^q$ for a fixed $q>0$, it suffices to take
\[
    m
    =
    O(K\log N(\log\log N)^3).
\]
The gain in the first term is asymptotically nontrivial whenever
$\log m=o(\log N)$.  In contrast, if $K=N^\alpha$ for a fixed
$0<\alpha<1$, then $\log m=\Theta(\log N)$ at the sufficient scale guaranteed
by the theorem, so our sufficient bound and the KMR bound have the same
logarithmic order.

The proof sharpens the coarse-scale Maurey estimate of~\cite{kmr14}; see
Table~\ref{tab:kmr_schatten_comparison} for a comparison.

\paragraph{Random row sets.}
The fixed-row and random-row models differ in the order of quantifiers.  In our
setting, $\Omega$ is prescribed before $g$ is drawn and may be any subset of
$\mathbb Z_N$ with $m$ elements.  All probability is over $g$, and the
probability estimate is required uniformly over every such $\Omega$.  To make
this distinction explicit, define
$q(\Omega):=\Pr_g[\delta_K(A_\Omega(g))>\delta]$.  The two guarantees are
\[
\sup_{\Omega\subset\mathbb Z_N,\ |\Omega|=m}q(\Omega)\leq\eta
\qquad\text{versus}\qquad
\E_{\Omega}[q(\Omega)]\leq\eta.
\]
In the second inequality, $\Omega$ is independent of $g$ and uniformly
distributed over the $m$-subsets of $\mathbb Z_N$.  The supremum inequality is
the arbitrary-fixed-row guarantee, whereas the expectation inequality is the
random-row guarantee.  At the same sample size and failure level, the former is
strictly stronger and correspondingly more demanding to establish: it controls
every prescribed, possibly highly structured sampling pattern without
averaging over the rows.  It does not, however, assert that a single draw of
$g$ succeeds simultaneously for all $\Omega$.

An arbitrary-fixed-row theorem immediately implies the corresponding guarantee
for a uniformly random $m$-subset without replacement, with no loss, by
averaging over $\Omega$.  The converse is false because
$\E_{\Omega}[q(\Omega)]\leq\eta$ does not control
$\sup_{\Omega}q(\Omega)$.  The average bound implies the existence of at least
one row set satisfying the same failure bound and, by Markov's inequality,
$\Pr_{\Omega}[q(\Omega)>t]\leq\eta/t$ for $t>0$.  For a prescribed set
$\Omega_0$ of positive sampling probability, it gives only the generally
vacuous estimate
$q(\Omega_0)\leq\eta/\Pr_{\Omega}[\Omega=\Omega_0]$.  Consequently, averaging
alone does not convert a random-row theorem into an arbitrary-fixed-row
theorem, even if one permits a universal multiplicative increase in $m$.

Krahmer--Mendelson--Rauhut~\cite[Theorem~4.1]{kmr14} study the same
arbitrary-fixed-row, standard two-sided $\ell_2/\ell_2$-RIP problem.  Their
$\Omega$ is an arbitrary fixed set of $m$ distinct row indices, and the only
randomness comes from a subgaussian generator with independent entries.  Our
matrix model, RIP notion, and order of quantifiers coincide with theirs; we
specialize the generator to a standard Gaussian vector and improve the sample
bound in this case.  The comparison with KMR is therefore direct and does not
rely on converting a random-row result.

By contrast, Huang--Pang--Xu~\cite[Theorem~1.1]{hpx19} draw $\Omega$ as a
multiset of $m$ independent uniform indices, so the rows are random and sampled
with replacement.  For fixed distortion and $K\lesssim N/\log^4N$, their
$m\gtrsim K\log^2(eK)\log N$ bound applies when the generator has i.i.d.
mean-zero, variance-one entries bounded by a fixed constant.  Their joint
probability estimate over the row multiset and the generator does not yield a
uniform estimate for an arbitrary prescribed $\Omega_0$ at the same failure
level and therefore does not establish this sample bound in our fixed-row
setting.  Huang--Pang--Xu~\cite[Remark~1.2]{hpx19} further observe that applying
their argument on the high-probability boundedness event for a Gaussian
generator yields the larger sample requirement
$m\gtrsim K\log^2(K\log N)\log^2N$, not the Gaussian bound proved here.  Their
results and ours therefore concern genuinely different probability models.

\paragraph{Modified circulant ensembles.}
In a related direction, Nelson--Price--Wootters~\cite[Theorems~6 and~9]{npw14}
obtain RIP matrices supporting fast matrix--vector multiplication under the
sufficient condition
$m\gtrsim\delta^{-2}K\log N\log^2(K\log N)$, using a subsampled circulant
matrix as an intermediate component.  Their construction, however, draws the
row multiset at random and aggregates the selected rows using fresh random
signs.  Consequently, the resulting ensemble has random rows and is no longer
a partial circulant matrix with a prescribed row set.

\paragraph{Fourier ensembles.}
A parallel line of work studies partial Fourier matrices, in which a
deterministic Fourier transform is restricted to a random set of rows.
Rudelson and Vershynin~\cite{rv08} obtained an early near-linear row bound,
up to polylogarithmic factors.  For fixed distortion,
Bourgain~\cite{bou14} established the sufficient row count
$m=O(K\log K\log^2N)$ for RIP of order $K$.  Haviv and
Regev~\cite[Theorem~4.5]{hr17} subsequently obtained the fixed-distortion
bound $m=O(K\log^2K\log N)$ for uniformly and independently sampled rows of
any unitary matrix with entries of magnitude $O(N^{-1/2})$; their general
distortion-dependent sufficient bound is
$m=O(\delta^{-2}\log^2(1/\delta)K\log^2(K/\delta)\log N)$.
On the lower-bound side, Bandeira, Lewis, and Mixon~\cite[Theorem~16]{blm18}
proved that, for sufficiently large $K$ dividing $N$ and fixed $\delta<1/3$, the
requirement $m=\Omega_\delta(K\log N)$ is necessary for the cyclic discrete
Fourier matrix to achieve $K$-RIP with probability at least $2/3$.  For the
Walsh matrix, equivalently the Fourier transform on $\mathbb F_2^n$,
B{\l}asiok et al.~\cite{bllmr23} showed that, when
$\min(K,N/K)>\log^C N$, the expected number of sampled rows must satisfy
$m=\Omega(K\log K\log(N/K))$: below this scale, independent Bernoulli row
sampling leaves a $K$-sparse vector in the kernel with probability $1-o(1)$.
This stronger Walsh lower bound is specific to the Walsh ensemble and does not
establish the same lower bound for the cyclic discrete Fourier matrix.  These
results do not directly transfer to the ensemble studied here: in the partial
Fourier model the transform is deterministic and the randomness lies in the
sampled row set, whereas we prescribe $\Omega$ arbitrarily and randomize the
Gaussian circulant generator.  Consequently, the Fourier results do not imply
an arbitrary-fixed-row RIP bound for our Gaussian partial circulant matrix.
\section{A Schatten-moment entropy bound for Gaussian partial circulant matrices}

\paragraph{Notation.}
For $x\in\mathbb R^N$ and $1\leq p\leq\infty$, $\|x\|_p$ denotes the $\ell_p$ norm, and $\|x\|_0$ denotes the number of nonzero coordinates of $x$. We write $e_j$ for the $j$-th standard basis vector of $\mathbb R^N$ and $B_1^N$ for the unit ball of $\ell_1^N$. For a set $T\subset\mathbb R^N$, $\operatorname{conv}(T)$ denotes its convex hull. For a matrix $B$, $\|B\|_{\mathrm{op}}:=\sup_{\|x\|_2\leq1}\|Bx\|_2$ denotes the operator norm, $\|B\|_F$ the Frobenius norm, and, for $p\geq1$, $\|B\|_{S_p}$ the Schatten $p$-norm, namely the $\ell_p$ norm of the vector of singular values of $B$. For a set $T$ equipped with a norm $\|\cdot\|$ and $u>0$, the covering number $\mathcal N(T,\|\cdot\|,u)$ is the minimal number of balls of radius $u$ whose union contains $T$. Throughout, $\varepsilon_1,\varepsilon_2,\ldots$ denote independent Rademacher random variables, uniform on $\{-1,+1\}$ and independent of all other randomness, and $C,c>0$ denote universal constants whose values may change at each occurrence.

\begin{table}[ht]
\centering
\caption{Comparison of the original KMR entropy argument with the
Schatten-moment refinement. We write
$V_z=m^{-1/2}R_\Omega C_z$ and $\Delta=\sqrt{K/m}$, and suppress the common
lower-order fine-scale term. The original KMR argument begins with the Fourier
bound $\|V_z\|_{\mathrm{op}}\leq m^{-1/2}\|Fz\|_\infty$, where
$F_{jk}=\exp(2\pi i jk/N)$ is the unnormalized discrete Fourier transform.}
\label{tab:kmr_schatten_comparison}
\begingroup
\small
\setlength{\tabcolsep}{4pt}
\renewcommand{\arraystretch}{1.2}
\begin{tabular}{p{0.18\textwidth}p{0.37\textwidth}p{0.37\textwidth}}
\hline
{\centering\textbf{Step}\par}
&
{\centering\textbf{Original KMR proof}\par}
&
{\centering\textbf{Schatten-moment refinement}\par}
\\
\hline
\textbf{Operator metric}
&
{\raggedright Use $\|V_z\|_{\mathrm{op}}\leq m^{-1/2}\|Fz\|_\infty$.\par}
&
{\raggedright Keep $V_z=m^{-1/2}R_\Omega C_z$ and work directly with the
partial-shift matrix.\par}
\\[2pt]
\textbf{Maurey sum}
&
{\raggedright Fourier maximum: $\sqrt{L\log N/m}$.\par}
&
{\raggedright Schatten estimate: $\sqrt{L\log(em)/m}$.\par}
\\[2pt]
\textbf{New logarithm}
&
{\raggedright Squaring the Fourier maximum contributes $\log N$.\par}
&
{\raggedright Squaring the Schatten estimate contributes $\log(em)$.\par}
\\[2pt]
\textbf{Atom count}
&
{\raggedright The $N$ coordinate atoms contribute $\log N$.\par}
&
{\raggedright Unchanged: the same atom-counting factor $\log N$ remains.\par}
\\
\hline
\textbf{Coarse entropy}
&
$\log\mathcal N\lesssim(\Delta/u)^2\log^2N$.
&
$\log\mathcal N\lesssim(\Delta/u)^2\log N\log(em)$.
\\[2pt]
\textbf{Chaining term}
&
$\gamma_2\lesssim\Delta\log(eK)\log N$.
&
$\gamma_2\lesssim\Delta\log(eK)\sqrt{\log(2N)\log(em)}$.
\\[2pt]
\textbf{RIP condition}
&
$m\gtrsim\delta^{-2}K\log^2(eK)\log^2(2N)$.
&
$m\gtrsim\delta^{-2}K\log^2(eK)\log(2N)\log(em)$.
\\
\hline
\end{tabular}
\endgroup
\end{table}

\begin{proof}[Proof of Theorem~\ref{thm:gaussian_upper}]
Let
\[
    \mathcal D_{K,N}
    :=
    \{
        x\in\mathbb R^N:
        \|x\|_2\leq 1,\ \|x\|_0\leq K
    \},
    \qquad
    V_x:=\frac{1}{\sqrt m}R_\Omega C_x,
\]
and write $\mathcal A:=\{V_x:x\in\mathcal D_{K,N}\}$.  By
Definition~\ref{def:partial_random_circulant} and commutativity of circular
convolution, $A_\Omega(g)x=V_xg$.  Since $|\Omega|=m$, the matrix $V_x$ consists of the $m$ rows of $\frac{1}{\sqrt m}C_x$ indexed by $\Omega$.  Therefore
\[
    \E_g[\|V_xg\|_2^2]
    =
    \|V_x\|_F^2
    =
    \|x\|_2^2.
\]
The first equality is the Gaussian second-moment identity
$\E_g[\|Bg\|_2^2]=\|B\|_F^2$.  The second holds because $V_x$ has $m$
rows, each with squared Euclidean norm $\|x\|_2^2/m$.  By
Definition~\ref{def:rip} and homogeneity, normalizing a nonzero
$x\in\mathcal D_{K,N}$ preserves its support, while the zero vector contributes zero to the supremum.  Thus
\[
    \delta_K(A_\Omega(g))
    =
    \sup_{V_x\in\mathcal A}
    |
        \|V_xg\|_2^2-\E_g[\|V_xg\|_2^2]
    |.
\]

The two radii entering the chaos estimate satisfy
\begin{equation}
\label{eq:gaussian_radii}
    d_F(\mathcal A)
    :=
    \sup_{V_x\in\mathcal A}\|V_x\|_F
    =1,
    \qquad
    d_{\mathrm{op}}(\mathcal A)
    :=
    \sup_{V_x\in\mathcal A}\|V_x\|_{\mathrm{op}}
    \leq
    \Delta,
    \quad
    \Delta:=\sqrt{ K/m}.
\end{equation}
For the Frobenius-radius equality, the preceding identity gives $\|V_x\|_F=\|x\|_2$, whose supremum over $\mathcal D_{K,N}$ is at most $1$ and is attained at $x=e_j$.  The operator-radius bound is obtained by taking the
supremum in the following three-step estimate:
\[
    \|V_x\|_{\mathrm{op}}
    \leq
    \frac{1}{\sqrt m}\|C_x\|_{\mathrm{op}}
    \leq
    \frac{1}{\sqrt m}\|x\|_1
    \leq
    \sqrt{\frac Km}\|x\|_2.
\]
The first inequality uses $\|R_\Omega\|_{\mathrm{op}}\leq1$; the second is Young's convolution inequality $\|x*y\|_2\leq\|x\|_1\|y\|_2$ on the finite group $\mathbb Z_N$ with counting measure \cite[Theorem~(20.18)]{hr79}; and the third is
Cauchy--Schwarz on the at most $K$ nonzero coordinates of $x$.
It remains to estimate
\[
    \gamma_2
    :=
    \gamma_2(\mathcal A,\|\cdot\|_{\mathrm{op}}).
\]

Here $\gamma_2$ denotes Talagrand's generic-chaining functional \cite[Chapter~2]{tal14}; we use only the entropy-integral upper bound recalled below.
Define the norm $\|z\|_X:=\|V_z\|_{\mathrm{op}}$ and let
\[
    U:=\{\pm e_j:j\in\mathbb Z_N\}.
\]
We have
\[
    \mathcal D_{K,N}
    \subset
    \sqrt K B_1^N
    =
    \sqrt K\operatorname{conv}(U).
\]
The inclusion uses the $K$-sparse Cauchy--Schwarz bound
$\|x\|_1\leq\sqrt K\|x\|_2$, together with $\|x\|_2\leq1$; the equality
is the standard identity $B_1^N=\operatorname{conv}(U)$.  Therefore we may use
Maurey's empirical method \cite{pis81} in the finite-atom form of
\cite[Lemma~4.2]{kmr14}, following Carl's entropy lemma
\cite[Proposition~3]{car85}.  The gain comes from the following Schatten
moment estimate.  Fix any $u_1,\ldots,u_L\in U$ and set
\[
    B_\ell:=R_\Omega C_{u_\ell},
    \qquad 1\leq\ell\leq L.
\]
Each $B_\ell$ is, up to sign, an $m\times N$ restriction of a cyclic
permutation matrix.  Thus
$
    \sum_{\ell=1}^L B_\ell B_\ell^*
    =
    LI_m,
$ 
whereas $\sum_{\ell=1}^L B_\ell^*B_\ell$ is diagonal. Indeed, each $B_\ell$ consists of $m$ distinct rows of a signed permutation matrix, so its rows are orthonormal and $B_\ell B_\ell^*=I_m$; moreover, $B_\ell^*B_\ell=C_{u_\ell}^*R_\Omega^*R_\Omega C_{u_\ell}$ arises from the zero--one diagonal projection $R_\Omega^*R_\Omega$ by conjugation with a signed permutation matrix, which only permutes the diagonal entries, so each $B_\ell^*B_\ell$ is a zero--one diagonal matrix with exactly $m$ ones. If
$c_1,\ldots,c_N$ are its diagonal entries, then
\[
    0\leq c_j\leq L,
    \qquad
    \sum_{j=1}^N c_j=mL.
\]
The two bounds hold because each $c_j$ is a sum of $L$ zero--one diagonal entries; the sum identity holds because each of the $L$ matrices $B_\ell^*B_\ell$ contributes exactly $m$ ones.  Consequently, for every $p\geq2$,
\[
    \|
        (\sum_{\ell=1}^L B_\ell B_\ell^*)^{1/2}
    \|_{S_p}
    =
    \sqrt L\,m^{1/p},
\]
and
\begin{align*}
    \|
        (\sum_{\ell=1}^L B_\ell^*B_\ell)^{1/2}
    \|_{S_p}^p
    &=
    \sum_{j=1}^N c_j^{p/2} 
    \leq
    L^{p/2-1}\sum_{j=1}^N c_j
    =
    L^{p/2}m.
\end{align*}
For the first display, the identity $\sum_{\ell=1}^L B_\ell B_\ell^*=LI_m$ gives $(\sum_{\ell=1}^L B_\ell B_\ell^*)^{1/2}=\sqrt L\,I_m$, whose $m$ singular values all equal $\sqrt L$, so its $S_p$ norm is $\sqrt L\,m^{1/p}$. In the second display, the first equality uses that $\sum_\ell B_\ell^*B_\ell$ is diagonal with
eigenvalues $c_j$.  The inequality follows term by term from
$c_j^{p/2}\leq L^{p/2-1}c_j$, since $p\geq2$ and $c_j\in[0,L]$; the final
equality uses $\sum_jc_j=mL$.  The Schatten-class noncommutative Khintchine inequality
\cite{lus86,lp91} now yields
\[
    (
        \E_\varepsilon[
        \|
            \sum_{\ell=1}^L\varepsilon_\ell B_\ell
        \|_{S_p}^p
        ]
    )^{1/p}
    \leq
    C\sqrt{pL}\,m^{1/p}.
\]
Taking $p=2\log(em)$, using $m^{1/p}\leq\sqrt e$,
$\|\cdot\|_{\mathrm{op}}\leq\|\cdot\|_{S_p}$, and Jensen's inequality gives
\begin{equation}
\label{eq:gaussian_nck}
    \E_\varepsilon[
    \|
        \sum_{\ell=1}^L\varepsilon_\ell B_\ell
    \|_{\mathrm{op}}
    ]
    \leq
    C\sqrt{L\log(em)}.
\end{equation}
The key point is that both square functions have $S_p$ norm at most $\sqrt L\,m^{1/p}$, so the choice $p\asymp\log(em)$ depends on $m$ rather than on the ambient dimension $N$.

For any $z\in\sqrt K\operatorname{conv}(U)$, choose independent
$U$-valued random vectors $Z_1,\ldots,Z_L$ with
$\E[Z_\ell]=z/\sqrt K$.  The standard symmetrization inequality is given in \cite[Lemma~6.3]{lt91}; see also
\cite[Appendix~A.2]{kmr14}.  Combining it with Eq.~\eqref{eq:gaussian_nck} gives
\[
\begin{aligned}
    \E[
    \|
        z-\frac{\sqrt K}{L}\sum_{\ell=1}^L Z_\ell
    \|_X
    ]
    &\leq
    \frac{2\sqrt K}{L}
    \E_{Z,\varepsilon}[
    \|
        \sum_{\ell=1}^L\varepsilon_\ell Z_\ell
    \|_X
    ]
    \leq
    C\sqrt{\frac{K\log(em)}{mL}}.
\end{aligned}
\]
The first inequality is the standard symmetrization bound.  For the second,
condition on $Z_1,\ldots,Z_L$, use
$\|\sum_\ell\varepsilon_\ell Z_\ell\|_X
=m^{-1/2}\|\sum_\ell\varepsilon_\ell R_\Omega C_{Z_\ell}\|_{\mathrm{op}}$,
and apply Eq.~\eqref{eq:gaussian_nck}.  Take
\[
    L
    =
    \lceil
        C\frac{K}{mu^2}\log(em)
    \rceil .
\]
Since $0<u\leq\Delta$ and $\log(em)\geq1$, we have $L\leq C(\Delta/u)^2\log(em)$; for sufficiently large $C$, the preceding
expectation is at most $u$, so some realization gives the required empirical
approximation.  Since there are at most $(2N)^L$ empirical averages, this proves the coarse entropy estimate
\begin{equation}
\label{eq:gaussian_coarse_entropy}
    \log\mathcal N(
        \mathcal D_{K,N},\|\cdot\|_X,u
    )
    \leq
    C( {\Delta}/{u})^2
    \log(2N)\log(em),
    \qquad 0<u\leq\Delta.
\end{equation}
If the covering centers are required to lie in $\mathcal D_{K,N}$, apply the construction at
radius $u/2$ and recenter every nonempty empirical ball at a point of
$\mathcal D_{K,N}$; this affects only the constant $C$.

At fine scales, cover each coordinate subspace separately.  On a fixed
support of cardinality at most $K$,
\[
    \|x-y\|_X
    \leq
    \Delta\|x-y\|_2.
\]
This Lipschitz bound follows from the three-step operator estimate above, applied to $x-y$, which is supported on the same coordinate set and therefore has at most $K$ nonzero coordinates.
A volumetric cover of the Euclidean unit ball, followed by a union bound over
the supports, gives
\begin{equation}
\label{eq:gaussian_fine_entropy}
    \log\mathcal N(
        \mathcal D_{K,N},\|\cdot\|_X,u
    )
    \leq
    K\log(eN/K)
    +
    K\log(1+C\Delta/u).
\end{equation}

Here the volumetric estimate covers the Euclidean unit ball of a fixed support with at most $(1+2\Delta/u)^K$ centers at $\ell_2$-radius $u/\Delta$ \cite[Proposition~C.3]{fr13}, and the union bound uses that the number of supports is at most $\sum_{k=0}^{K}\binom{N}{k}\leq(eN/K)^K$.

Dudley's entropy bound \cite{dud67}, using
Eq.~\eqref{eq:gaussian_fine_entropy} below $u_0=m^{-1/2}$ and
Eq.~\eqref{eq:gaussian_coarse_entropy} above $u_0$, gives
\[
\begin{aligned}
    \gamma_2
    &\leq
    C\int_0^{u_0}
        \sqrt{
            K\log(eN/K)
            +
            K\log(1+C\Delta/u)
        } \d u  +
    C\int_{u_0}^{\Delta}
        \frac{\Delta}{u}
        \sqrt{\log(2N)\log(em)} \d u.
\end{aligned}
\]
The first integral is at most
\[
    C\Delta
    (
        \sqrt{\log(eN/K)}
        +
        \sqrt{\log(eK)}
    ),
\]
and the second is at most
\[
    C\Delta\log(eK)\sqrt{\log(2N)\log(em)}.
\]
The first estimate uses $u_0=\Delta/\sqrt K$ and
$\int_0^1\sqrt{\log(1+C\sqrt K/v)}\d v\lesssim\sqrt{\log(eK)}$; the second
uses $\log(\Delta/u_0)=\frac12\log K$.
We therefore obtain
\begin{equation}
\label{eq:gaussian_gamma}
    \gamma_2
    \leq
    C\sqrt{\frac Km}
    (
        \sqrt{\log(eN/K)}
        +
        \log(eK)\sqrt{\log(2N)\log(em)}
    ).
\end{equation}

Since the entries of $g$ are independent standard Gaussian random variables, the
chaos estimate of Krahmer--Mendelson--Rauhut
\cite[Theorem~3.1]{kmr14} applies.  Set
\[
    E:=\gamma_2(\gamma_2+d_F)+d_Fd_{\mathrm{op}},
    \qquad
    V:=d_{\mathrm{op}}(\gamma_2+d_F),
    \qquad
    W:=d_{\mathrm{op}}^2.
\]
For every $t>0$, the chaos estimate gives
\[
    \Pr_g[
        \delta_K(A_\Omega(g))\geq CE+t
    ]
    \leq
    2\exp(
        -c\min\{
            \frac{t^2}{V^2},
            \frac{t}{W}
        \}
    ).
\]
The first term in Eq.~\eqref{eq:gaussian_sample_implicit}, together with
Eq.~\eqref{eq:gaussian_radii} and Eq.~\eqref{eq:gaussian_gamma}, implies
\[
    \gamma_2\leq c_0\delta,
    \qquad
    d_{\mathrm{op}}\leq c_0\delta,
\]
after increasing the universal constant in
Eq.~\eqref{eq:gaussian_sample_implicit}. Here $c_0\in(0,1)$ is a sufficiently small universal constant. For the first bound, substitute the
sample condition into Eq.~\eqref{eq:gaussian_gamma} and use
$\log(eN/K)\leq2\log(2N)$; the second follows from
$d_{\mathrm{op}}\leq\sqrt{K/m}$ in Eq.~\eqref{eq:gaussian_radii}.  Since
$d_F=1$ and $\delta\leq1$, increasing the universal constant further gives
\[
    CE\leq\delta/2,
    \qquad
    V\leq C\sqrt{K/m},
    \qquad
    W\leq K/m.
\]
The first bound follows from
$E=\gamma_2(\gamma_2+1)+d_{\mathrm{op}}$; the other two follow from the definitions of $V$ and $W$ and Eq.~\eqref{eq:gaussian_radii}.  Taking $t=\delta/2$
and using
$\{\delta_K(A_\Omega(g))>\delta\}\subseteq
\{\delta_K(A_\Omega(g))\geq CE+t\}$ gives
\[
    \Pr_g[
        \delta_K(A_\Omega(g))>\delta
    ]
    \leq
    2\exp(-c\frac{\delta^2m}{K})
    \leq
    \eta.
\]
The first inequality follows from $t^2/V^2\geq c\delta^2m/K$ and
$t/W\geq c\delta m/K\geq c\delta^2m/K$.  The second term in
Eq.~\eqref{eq:gaussian_sample_implicit} gives the last inequality after increasing
the universal constant.  This completes the proof.
\end{proof}



\section*{Acknowledgments}

The author used Gemini Pro 1.3, Codex Sol 5.5, and Claude Code Fable 5 in preparing this paper, specifically to assist with literature searches, grammatical review, and language editing. The author is also grateful to Eric Price, Jelani Nelson, and David Woodruff for helpful discussions on this topic during the author's doctoral studies.

\bibliographystyle{alpha}
\bibliography{ref}

\end{document}